\pdfoutput=1 
\documentclass[sigconf]{acmart}
\usepackage{todonotes}
\usepackage{centernot}
\usepackage{todonotes}
\usepackage{amsmath}
\usepackage{mathtools}
\usepackage[noend]{algpseudocode}
\usepackage{algorithm}
\usepackage{amsfonts}
\usepackage{mathrsfs}
\usepackage{comment}
\AtBeginDocument{%
  }

\setcopyright{acmcopyright}
\copyrightyear{2023}
\acmYear{2023}
\acmDOI{XXXXXXX.XXXXXXX}

\acmConference[MobiHoc '23]{Submitted}{June 03--05,
  2023}{Washington, DC}
\acmPrice{15.00}
\acmISBN{978-1-4503-XXXX-X/18/06}

\begin{document}
\def\E{\mathbb E}
\def\tloc{T}
\def\noise{\varrho}
\def\edi#1{{\color{black}#1}}
\def\beh#1{{\color{blue}#1}}
\def\red#1{{{\color{red}#1}}}
\def\roh#1{{{\color{red}#1}}}
\def\fskip#1{}
\def\dmin{d_{\min}^+}
\def\dmax{d_{\max}^+}
\def\P{\mathbb P}
\def\G{\mathcal G}
\def\sing{\sigma}
\def\pstar{\mathcal P^*}
\def\R{\mathbb{R}}
\def\boldelta{\mathbf\Delta}
\def\N{\mathbb N}
\def\allone{\mathbf 1}
\def\sgd{\widetilde \nabla}
\newcommand{\norm}[1]{\left\lVert#1\right\rVert}
\newcommand{\doubleperp}{\perp\!\!\!\perp}
\title{Connectivity-Aware Semi-Decentralized Federated Learning over Time-Varying D2D Networks}

\author{Rohit Parasnis}
\authornotemark[1]
\affiliation{%
  \institution{Purdue University}
  \streetaddress{610 Purdue Mall}
  \city{West Lafayette}
  \state{Indiana}
  \country{USA}
  \postcode{47907-2040}
}
\email{rparasni@purdue.edu}

\author{Seyyedali Hosseinalipour}
\affiliation{%
  \institution{University at Buffalo--SUNY}
  \streetaddress{12 Capen Hall}
  \city{Buffalo}
  \state{New York}
  \country{USA}}
\email{alipour@buffalo.edu}

\author{Yun-Wei Chu}
\affiliation{%
  \institution{Purdue University}
  \streetaddress{610 Purdue Mall}
  \city{West Lafayette}
  \state{Indiana}
  \country{USA}
  \postcode{47907-2040}
}
\email{chu198@purdue.edu}

\author{Mung Chiang}
\affiliation{%
  \institution{Purdue University}
  \city{West Lafayette}
  \country{USA}}
\email{chiang@purdue.edu}

\author{Christopher G. Brinton}
\affiliation{%
  \institution{Purdue University}
  \streetaddress{610 Purdue Mall}
  \city{West Lafayette}
  \country{USA}}
\email{cgb@purdue.edu}

\renewcommand{\shortauthors}{Parasnis, Hosseinalipour, Chu, Brinton, Chiang.}

\begin{abstract}
  Semi-decentralized federated learning blends the conventional device-to-server (D2S) interaction structure of federated model training with localized device-to-device (D2D) communications. We study this architecture over practical edge networks with multiple D2D clusters modeled as time-varying and directed communication graphs. Our investigation results in an algorithm that controls the fundamental trade-off between (a) the rate of convergence of the model training process towards the global optimizer, and (b) the number of D2S transmissions required for global aggregation. Specifically, in our semi-decentralized methodology, D2D consensus updates are injected into the federated averaging framework based on column-stochastic weight matrices that encapsulate the connectivity within the clusters. To arrive at our algorithm, we show how the expected optimality gap in the current global model depends on the greatest two singular values of the weighted adjacency matrices (and hence on the densities) of the D2D clusters. We then derive tight bounds on these singular values in terms of the node degrees of the D2D clusters, and we use the resulting expressions to design a threshold on the number of clients required to participate in any given global aggregation round so as to ensure a desired convergence rate. Simulations performed on real-world datasets reveal that our connectivity-aware algorithm reduces the total communication cost required to reach a target accuracy significantly compared with baselines depending on the connectivity structure and the learning task.
\end{abstract}




\maketitle

\section{Introduction}
Federated learning (FL)~\cite{konevcny2016federated, mcmahan2017communication} is a popular paradigm for distributing machine learning (ML) tasks over a network of centrally coordinated devices. By not requiring the devices to share any training data with the central coordinator (server), FL improves privacy and communication efficiency. The first FL technique, known as federated averaging (FedAvg), was proposed in~\cite{konevcny2016federated, mcmahan2017communication} as a distributed optimization algorithm for a ``star'' topology-based network architecture. In each iteration of the FedAvg algorithm, (i) devices individually performs a number of local stochastic gradient descent (SGD) iterations and transmit their cumulative stochastic gradients to the central server, which then (ii) aggregates a random subset of these gradients to estimate the globally optimal ML model.  In recent years, several variants of FedAvg have been proposed to address the challenges encountered by FL at the wireless edge, including different dimensions of heterogeneity in dataset statistics (e.g., varying local data distributions) and in the network system itself (e.g., varying communication and computation capabilities). 

An emerging arch of work has been exploring FL under edge networks that diverge from the star learning topology between the devices and the server. This had led to varying degrees of decentralization in FL, reaching fully decentralized, serverless settings that sit at the opposite extreme of the star topology~\cite{lalitha2018fully,sun2022decentralized,zehtabi2022event, beltran2022decentralized,hua2022efficient,li2020federated, lin2021federated}. In between these two extremes is \textit{semi-decentralized FL}, where device-to-device (D2D) communications complement device-to-server (D2S) interactions~\cite{lin2021semi,yemini2022semi,hosseinalipour2022multi,briggs2020federated}. These D2D interactions occur locally within \textit{clusters} of devices, with each cluster forming a connected component. In semi-decentralized FL, D2D transmissions are less energy-consuming than D2S interactions and can help reduce the frequency of D2S communications through localized synchronizations of the ML model updates. 


Despite these recent investigations, we still do not have a clear understanding of how different D2D topology properties impact the learning process. For instance, the ratio of the number of D2D interactions to that of D2S interactions will impact the training efficiency differently over different topologies. This becomes especially important in the presence of constraints such as upload/download bandwidths, and stochastic uncertainties such as data heterogeneity, client mobility, and communication link failures.
\begin{figure}
    \centering
    \includegraphics[scale=0.35]{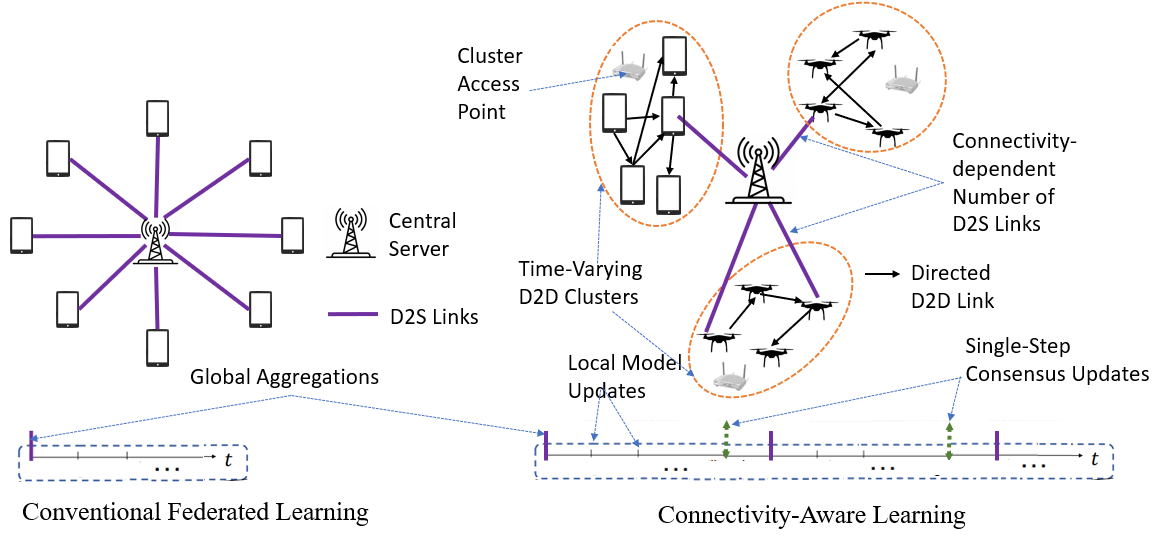}
    \caption{Conventional Federated Learning vs. Connectivity-Aware Semi-Decentralized Learning Architecture}
    \label{fig:architecture}
\end{figure}
On one hand, edge devices in clustered D2D networks that have little to no cross-cluster interactions
are typically in contact with only a small fraction of the rest of the network at any given time instant (e.g., networks of unmanned aerial vehicle (UAV) swarms spread over geo-distributed regions separated by long distances). In such networks, if there is no central coordinator (implying zero D2S interactions) and if the training data are distributed heteregeneously among the edge devices, no practically feasible number of D2D interactions
is likely to aggregate a set of local ML models that are diverse enough to approximate the global data distribution
~\cite{bellet2022d}.

On the other hand, having a high number of D2D interactions is advantageous when D2S interactions take the form of high-latency, high-energy transmissions (e.g., if the UAV swarms in the previous example are miles away from the nearest base station). Moreover, classical star-topology-based FL architectures miss out on an important benefit
of D2D cooperation: devices acting as information relays between other devices and the server, effectively sharing with the server more information than it would expect to receive. 

We are thus motivated to conduct a formal study of semi-decentralized FL, and reveal the combined impact of D2S and D2D interactions on the training process. After building an understanding of the D2D topologies on which the D2D interactions occur, we propose a novel FL technique that enables us to take into account the degree distributions of the D2D clusters and use this knowledge to tune the number of expensive D2S transmissions  while simultaneously ensuring a minimum rate of global training convergence. As shown in Fig.~\ref{fig:architecture}, we incorporate two scales of model aggregations: on the first scale, the edge devices perform intra-cluster model aggregations with their one-hop neighbors via distributed averaging, and on the second scale, a central server samples a random set of clients (as in the classical FedAvg architecture~\cite{mcmahan2017communication}) for global aggregation. 



Our methodology has several potential use-cases, including the following that we will refer to as examples throughout the paper:
\begin{enumerate}
    \item \textbf{\textit{UAV Networks for ISR:}} UAVs are being increasingly deployed for intelligence, surveillance, and reconnaissance (ISR) operations in defense settings~\cite{wang2017uncertain,shen2008game}. With the UAVs partitioned into D2D-enabled swarms deployed across different areas, our connectivity-aware algorithm can facilitate intra-cluster communications and reduce the over-reliance of the model training process on D2S transmissions.
    \item\textbf{\textit{Self-driving cars:}} Many learning tasks for self-driving cars call for vehicles to communicate over short distances. In such settings, geographical proximity can be used to partition the traffic network into clusters, which would enable us to design intra-cluster D2S communications that turn out to be more efficient than D2S communications with a far-away server.
\end{enumerate}
\subsection{Summary of Contributions}
Our key contributions are summarized below:
\begin{enumerate}
    \item \textbf{\textit{Analysis with Time-Varying and Directed Cluster Topologies:}} We consider that each D2D cluster in general is a time-varying directed graph (digraph). We show how the expected optimality gap of the learning process depends on the greatest two singular values of the weighted adjacency matrices used for local aggregations in the clusters. Our analysis is applicable to edge networks with asymmetric D2D communications subjected to  link failures.
    \item \textbf{\textit{Singular Value Bounds in terms of Node Degrees:}} 
    We derive bounds on the  singular values of the cluster-specific weighted adjacency matrices in terms of the degree distribution of every cluster. This introduces new technical challenges as described in Section 1.2, since it is a stark departure from existing analyses of consensus-based FL algorithms that rely heavily on the spectral gaps of symmetric weight matrices (e.g., see~\cite{koloskova2021improved,lin2021semi,xing2021federated,koloskova2020unified,beznosikov2021decentralized,zehtabi2022event,hosseinalipour2022multi}).
    \item \textbf{\textit{Connectivity-Aware Learning Algorithm:}} We
    use our singular value bounds to design a time-varying  threshold on the number
    of clients required to be sampled by the central server for global
    aggregation so as to enforce a desired convergence rate while simultaneously reducing the number of D2S communications. This tradeoff results in a novel connectivity-aware algorithm with significant energy savings, as validated subsequently by our numerical results.
    \item \textbf{\textit{Effect of Data Heterogeneity under Mild Gradient Diversity Assumptions:}} We derive a bound on the expected optimality gap that captures the effects of cluster densities as well as the extent of data heterogeneity across the devices. In doing so, we employ a milder definition of gradient diversity~\cite{lin2021semi} than what is typically assumed in literature.
\end{enumerate}

\textbf{\textit{Notation:}} We denote the set of real numbers by $\R$ and the set of positive integers by $\N$. For any $n\in\N$, we define $[n]:=\{1,2,\ldots, n\}$. For a finite set $S$, we denote its cardinality by $|S|$.
  
  We denote the vector space of $n$-dimensional real-valued column vectors by $\R^n$. We use the superscript notation $^\top$ to denote the transpose of a vector or a matrix. All matrix and vector inequalities are assumed to hold entry-wise. We use $I$ to denote the identity matrix (of the known dimension) and $\allone$ to denote the column vector (of the known dimension) that has all entries equal to 1. Similarly, $\mathbf 0$ denotes the all-zeroes vector. In addition, we use $\|\cdot\|$ to denote the Euclidean norm of a square matrix or a vector, and for any vector $v\in\R^n$ we use $\text{diag}(v)$ to denote the diagonal matrix whose $i$-th diagonal entry is $v_i$.  
  
  We say that a vector $v\in\R^n$ is \textit{stochastic} if $v\geq 0$ and $v^\top\allone = 1$, and a matrix $A$ is \textit{column-stochastic} if $A$ is non-negative and if each column of $A$ sums to 1, i.e., if $A\geq 0$ and $A^{\top}\allone=\allone$. 

\subsection{Related Work}

 
Several different FL approaches with varying degrees of decentralization have been proposed to date. In this section, we focus on those which are most relevant to the present work.

 \textit{\textbf{Semi-decentralized FL:}}~\cite{lin2021semi} proposes a semi-decentralized learning methodology in which the D2D network is partitioned into clusters, as in our paper. 
 The key differences between~\cite{lin2021semi} and the present work are (a) we do not assume the D2D communications to be bidirectional (equivalently, the cluster graphs in our model are not undirected), and (b) our analysis uses column-stochastic consensus matrices that need not satisfy the standard but unrealistic assumption of symmetry (which leads to double stochasticity and may not hold if the cluster graphs are directed). This leads to two significant technical challenges. First, we cannot use standard eigenvalue results in our analysis since we must focus on singular values, which generally differ from eigenvalues for asymmetric matrices. Second, unlike doubly stochastic matrices, column-stochastic aggregation matrices in general do not ensure convergence to consensus in the absence of a central coordinator, which means our analysis must account for the combined effect of global aggregations and column-stochasticity. We address these challenges in this work.
 
 Another closely related semi-decentralized learning methodology is~\cite{yemini2022semi}. In~\cite{yemini2022semi}, the goal is to enable edge devices to compute weighted sums of their neighbors' scaled cumulative gradients in order to reduce the dependence of the global training process on unreliable D2S links.~\cite{yemini2022semi}, however, assumes the D2D communication network to be time-invariant and undirected, thereby disregarding potential communication link failures and client mobility. 
 

 \textit{\textbf{Learning over Clustered D2D Networks:}} Recently,~\cite{al2022decentralized} proposed fully decentralized learning over D2D networks in which a small subset of nodes act as bridges between different clusters for cross-cluster model transmission, thereby obviating the role of a server. Their topology design, however, results in a static rather than a dynamic D2D network. Reference~\cite{briggs2020federated} also focuses on clustered networks, but it provides a semi-decentralized learning methodology where the basis for clustering is data similarity, whereas our methodology makes no assumptions on the basis for clustering. A complementary approach is proposed in~\cite{bellet2022d}, where every cluster is assumed to be a clique and the D2D network is partitioned in such a way that each local dataset is representative of the global data. Network clusters also form the focus of another recent work,~\cite{wang2022confederated}, which proposes having one edge server per cluster so as to eliminate the need for a central server. Its learning algorithm assumes the edge network topology to be undirected, which gives rise to a symmetric adjacency matrix.

 \textit{\textbf{Other Consensus-based Algorithms:} } Reference~\cite{koloskova2021improved} provides improved bounds on the convergence rates of certain gradient tracking methods used in decentralized learning by enhancing the analysis of the consensus matrix (referred to as the \textit{mixing matrix} therein) and its spectral gap. However, similar to ~\cite{lin2021semi}, this work assumes the consensus matrix to be row-stochastic as well as symmetric, and hence, doubly stochastic.  
 In this respect,~\cite{he2019central} relaxes the assumptions of symmetry as well as double stochasticity in an online learning setting. However, the matrices therein are row-stochastic, which are not average-preserving and hence, they are not as suitable as column-stochastic matrices for minimizing the average of all the local loss functions. Finally, we remark that there exists abundant literature on distributed optimization over time-varying digraphs characterized by consensus matrices that are not necessarily doubly stochastic (e.g., see~\cite{liang2019dual,xin2018linear,akbari2015distributed,nedic2017achieving,nedic2014distributed,nedic2016stochastic,wang2019edge}). However, the effects of both data heterogeneity (or non-i.i.d. data distribution) and graph-theoretic properties (such as the degree distribution of the network in question) on the convergence rate of these algorithms have remained largely unexplored.

\section{Semi-Decentralized FL Setup}\label{sec:prob_form}
We now introduce the system model, the learning objective, and the network model in semi-decentralized FL.

\subsection{System Model and Learning Objectives}
We consider a collaborative learning environment consisting of $n$ edge devices, or \textit{clients}, and a central parameter server (PS) that is tasked with aggregating all the local model updates generated by the clients. We use $[n]$ to denote the set of clients.

 Each client $i\in[n]$ has a local dataset $\mathcal D_i$, which is a collection of data samples of the form $\xi=(u,y)$ where $u\in\R^p$ is the \textit{feature vector} of the sample and $y$ is its \textit{label}. On this basis, for any model $x\in\R^p$, we define the \textit{loss function} $L:\R^p\times \cup_{i=1}^n \mathcal D_i\to\R$ so that $ L(x;\xi)$  denotes the loss  incurred by $x$ on a sample $\xi\in \cup_{i=1}^n \mathcal D_i$ (where $\cup_{i=1}^n\mathcal D_i$ is the global dataset). The average loss incurred by $x$ over the local dataset of client $i$ is given by  $
    f_i(x):=\frac{1 }{|\mathcal D_i|} \sum_{\xi \in \mathcal D_i} L(x;\xi),
 $
 where $f_i:\R^p\to \R$ denotes the \textit{local loss function} of client $i$.
 
 In collaboration with the PS, the clients seek to minimize the \textit{global loss function} $f:\R^p\to\R$, defined as the unweighted arithmetic mean
 $
    f(x):=\frac{1}{n}\sum_{i=1}^n f_i(x)
 $
  of all the local loss functions. The learning objective, therefore, is to determine the \textit{global optimum} $x^*:= \arg\min_{x\in\R^p} f(x)$.
 
 \subsection{D2D and D2S Network Models}\def\dinmax{d_{\max}^{-}}
 We model two types of interactions among the network elements: (i) D2S and (ii) D2D. For D2S interactions, the devices can engage in uplink communications to the PS if prompted by the server, which happens through a sampling procedure explained later.
 
 We model the D2D network as a time-varying directed graph ${G(t)=([n],E(t))}$, where $[n]$ denotes the \textit{vertex set} and $E(t)$ the \textit{edge set} of the digraph. The existence of a directed edge from a node $i\in[n]$ to another node $j\in[n]$ in $G(t)$ denotes the existence of a communication link from the $i$-th client to the $j$-th client in the D2D network. In this case, we refer to client $i$ (respectively, client $j$) as the in-neighbor (respectively, out-neighbor) of client $j$ (respectively, client $i$). The set of in-neighbors (respectively, out-neighbors) of a client $i\in [n]$ at time $t$ is denoted by $\mathcal N_i^-(t)$ (respectively, $\mathcal N_i^+(t)$). The number of in-neighbors (respectively, out-neighbors) is called the in-degree (respectively, out-degree) and is denoted by $d_i^-(t)$ (respectively, $d_i^+(t)$). We let $\dinmax(t)$, $\dmin(t)$, and $\dmax(t)$ denote the maximum in-degree, the minimum out-degree, and the maximum out-degree, respectively. 
 
 Unlike standard works on distributed learning~\cite{nedic2014distributed,nedic2016stochastic,nedic2017achieving,akbari2015distributed}, we do not assume the D2D network to be strongly connected or even uniformly strongly connected~\cite{nedic2014distributed,nedic2016stochastic} over time. This gives rise to a number $c>1$ of strongly connected components of $G(t)$, denoted $\{(V_1(t),E_1(t)),(V_2(t),E_2(t)),\ldots, (V_c(t),E_c(t))\}$ which we refer to as  \textit{clusters} of the D2D network. Here, we make the following mild assumptions that apply to many cellular networks:
 \begin{enumerate}
     \item The number of clusters, $c$, is time-invariant.
     \item There does not exist any communication link between any two clusters. In other words, $E(t)=\cup_{\ell=1}^c E_\ell(t)$.
     \item Regardless of any movement of clients from one cluster to another over time, 
     as of time $t$, the server has full knowledge of the vertex sets $\{V_\ell(t)\}_{\ell=1}^c$ of all the $c$ clusters.
 \end{enumerate}
 The third condition is satisfied in practice since the base station (which acts as the PS) is aware of the users in its coverage area.

\section{Proposed Method for Connectivity-Aware Learning}
 We now present our methodology for connectivity-aware learning over the semi-decentralized setup from Sec. 2. Our technique will enable the central server to use limited knowledge of the cluster degree distributions to tune a communication-efficiency trade-off. 
 
 \subsection{Local Model Updates}
 As in many FL schemes, we assume every client performs multiple rounds of local SGD iterations between any two consecutive rounds of global aggregation. Let $x^{(t)}$ denote the global model that all the clients possess at the end of the $t$-th round of global aggregation. Then, each client $i\in[n]$ performs $\tloc\in\N$ iterations of local SGD. In other words, for each $k\in\{0,1,\ldots, \tloc-1\}$, we have
 \begin{align}
     x_i^{(t,k+1)} = x_i^{(t,k)} - \eta_t\sgd f_i(x_i^{(t,k)}),
 \end{align}
 where $\eta_t>0$ is the learning rate or the step-size, and $\sgd f_i(x) :=\frac{1}{|\chi_i|} \sum_{\xi\in \chi_i} 
 \nabla L(x;\xi)$ is the stochastic gradient computed by client $i$ by sampling a  \textit{mini-batch} or a random subset $\chi_i\subset\mathcal D_i$ of its local samples. Note that $x_i^{(t,0)}:=x^{(t)}$.

 \subsection{Intra-Cluster Model Aggregations}     
 The next step involves all the clients aggregating their scaled cumulative gradients with their neighbors.
 This aggregation takes the form of  weighted sums.
 Every client $i\in[n]$ first transmits its\textit{ scaled cumulative stochastic gradient} 
 $x_i^{(t,T)}-x^{(t)} = -\eta_t\sum_{k=0}^{T-1} \sgd f_i(x_i^{(t,k)})$ to each of its out-neighbors $j\in\mathcal N_i^+(t)$ before the $t$-th global aggregation round.
 To facilitate this, we assume that every cluster $\ell\in[c]$ contains an access point  to which every client $i\in V_\ell(t)$ sends a list of its in-neighbors (clients whose gradients $i$ has received). The access point then announces the end of the concerned D2D communication round, determines the out-degree sequence $\{d_j^+(t):j\in V_\ell(t)\}$ of the cluster, and broadcasts this sequence to every client in the cluster.
 
 Subsequently, the client computes the following weighted sum of all the scaled cumulative gradients it receives from its in-neighbors:
 \begin{align}
    \Delta_i(t) = \sum_{j\in\mathcal N_i^{-}(t)} \frac{1 }{d_j^+(t) }\left(x_j^{(t,T)}-x^{(t)} \right).
 \end{align}
 \newtheorem{fact}{Fact}
 This rule can be expressed compactly in matrix form  as 
 \begin{align}
     \boldelta(t) = A(t)X_{\text{diff} }^\top(t),
 \end{align}
 where $\boldelta(t) := \left[\Delta_1(t)\,\,\,\Delta_2(t)\,\,\,\cdots\,\,\,\Delta_n(t)\right]^\top$,\\ $X_{\text{diff} }(t):=\left[x_1^{(t,T)}-x^{(t)}\,\,\, x_2^{(t,T)}-x^{(t)}\,\,\,\cdots\,\,\,x_n^{(t,T)}-x^{(t)} \right]$, and \\${A(t)\in \R^{n\times n}}$ is a matrix whose $(i,j)$-th entry equals $a_{ij}(t)=\frac{1}{d_j^+(t)}$ for all $i\in[n]$ and $j\in \mathcal N_i^-(t)$. 
 
 \begin{fact}
    $A(t)$ is a column-stochastic matrix because the following holds for all $j\in[n]$:
 $$
    \sum_{i=1}^n a_{ij}(t) = \sum_{i\in [n]: j\in \mathcal N_i^-(t)}\frac{1 }{ d_j^+(t) }= \sum_{i\in\mathcal N_j^+(t)} \frac{1}{|\mathcal N_j^+(t) |} = 1.
 $$
 It can be verified that $A(t)$ is a block-diagonal matrix with its blocks $\{A_\ell(t)\}_{\ell=1}^c$ being the equal-neighbor adjacency matrices of the $c$ clusters in the D2D network.
 \end{fact}
 Henceforth, we refer to $A(t)$ as the \textit{equal-neighbor adjacency matrix} of $G(t)$ because it represents every client $i\in[n]$ transmitting an equal share (a fraction $\frac{1}{d_i^+(t)}$) of its scaled cumulative gradient to its $d_i^+(t)$ out-neighbors.
 
 \subsection{Global Aggregation at the PS}
 For the global aggregation step, the PS samples a random subset of clients $\mathcal S(t)\subset [n]$. The cardinality $m(t)\le n$ 
 of this set is carefully chosen by our algorithm such that the resulting number of D2S interactions is just enough to complement the intra-cluster aggregations without excessively slowing down the training process. 
 
 Specifically, this involves three broad steps: (a) The PS first learns the degree distribution of each cluster. (b) It then computes an upper bound on an error quantity $\phi(t)$ that captures the combined effect of random sampling and the cluster degree distributions on the convergence rate. (c) It computes the minimum value of $m(t)$ required to keep $\phi(t)$ below a desired threshold. More specifically:
 \begin{enumerate}
     \item For the $(t+1)$-th round of global aggregation, the server uses $m(t)$ (computed in the previous iteration) to select $\left\lceil \left(\frac{m(t) }{n}\right)n_\ell(t)\right\rceil$ clients uniformly at random from the $n_\ell(t):=|V_\ell(t)|$ clients that constitute cluster $\ell\in[c]$. This ensures that every cluster has a representation in the global aggregation that is proportionate to its size. The resulting set of randomly sampled clients is denoted  by $\mathcal S(t)$. The server then updates the global model as follows:
     \begin{align}\label{eq:global_aggregation}
         x^{(t+1)}= x^{(t)}+ \frac{1}{m(t)}\sum_{i\in\mathcal S(t)} \Delta_i(t) =x^{(t)}+ \frac{1}{m(t)}\sum_{i=1}^n\tau_i(t)\Delta_i(t),
     \end{align}
     where $\tau_i(t):=|\{i\}\cap\mathcal S(t)|$ is an indicator random variable that takes the value $1$ when client $i$ is sampled and the value $0$ otherwise. Note that $\sum_{i=1}^n \tau_i(t)=|\mathcal S(t)|=m(t)$.
     \item The current round is now $t \leftarrow t+1$. All the cluster access points  send their respective out-degree sequences to the server. Using this information, the server computes $\alpha_\ell(t):=\frac{1}{n_\ell(t)}\min_{i\in V_\ell(t)}d_i^+(t)$, the \textit{minimum out-degree fraction} of cluster $\ell\in[c]$. The server then uses either of the two sets of singular value bounds that we later derive in Sec. 5 (either~\eqref{eq:first_sing_bound}-\eqref{eq:second_sing_bound} or~\eqref{eq:general_first_sing_bound}-\eqref{eq:general_second_sing_bound}) to compute an upper bound $\psi(m(t),\alpha_1(t),\ldots, \alpha_c(t))$ on the \textit{connectivity factor} affecting the convergence rate. This connectivity factor is defined as 
     \begin{align}\label{eq:conn_factor}
         \phi(t):=\left(\frac{n}{m(t)}-1 \right)\sum_{\ell=1}^c\frac{n_\ell(t)}{n}\phi_\ell(t),
     \end{align}
     where $\phi_\ell(t):=\sing_1^2(A_\ell(t)) + \sing_2^2(A_\ell(t))-1$ depends on the greatest two singular values $\sing_1(A_{\ell}(t))\ge\sing_2(A_{\ell}(t))$ of the equal-neighbor adjacency matrix $A_\ell(t)$ of cluster $\ell$. For the upper bound, we will show that
     \begin{align}
     \psi(m(t),\alpha_1(t),\ldots,\alpha_c(t))=\left(\frac{n}{m(t)}-1\right)\sum_{\ell=1}^c \frac{n_\ell(t)}{n}\psi_\ell(t), 
     \end{align}
     where either of the following holds (with the indexing $(t)$ on the right hand side omitted for brevity):
     \begin{align}
         \psi_\ell(t) &= 1+\varepsilon_\ell+\left(\frac{1}{\alpha_\ell}-1\right)^2 + 2\varepsilon_\ell\left(1 +\frac{2}{\alpha_\ell} - \frac{1}{\alpha_\ell^2} \right),\cr
         \psi_\ell(t) &= 2+2\varphi_\ell\cr
         &\quad-\frac{(1-\varepsilon_\ell)^2(1-\alpha_{-\ell}^2)\left((1-\varepsilon_\ell)^2(1-\alpha_{-\ell}^2)-\alpha_{-\ell}\right)}{n_\ell(\varepsilon_{\text{net,}\ell}+1)\left(\varepsilon_{\text{net,}\ell} - \alpha_{-\ell} + \frac{1}{\alpha_\ell n_\ell} \right )}
     \end{align}
     with $\varepsilon_\ell(t):=\frac{\dmax(t) - \dmin(t) }{\dmin(t)}$, $\varphi_\ell(t):=\frac{\dinmax(t) - \dmin (t)}{\dmin(t)}$, \\
     $\alpha_{-\ell}(t):=\frac{1}{\alpha_\ell(t)}-1$ and $\varepsilon_{\text{net,} \ell}(t)=\varphi_\ell(t) + \frac{\varepsilon_\ell(t)}{\alpha_\ell(t)}$.
     \item Finally, the server sets 
     \begin{align*}
         m(t+1):=
         \min\left\{r\in [n]: \psi(r, \alpha_1(t+1),\ldots, \alpha_c(t+1)) \le \phi_{\max}\right \}
     \end{align*}
     where $\phi_{\max}$ is a threshold given as an input to the algorithm. 
     This step ensures that  $\phi(t)$ remains below the threshold $\phi_{\max}$, thereby preserving the convergence rate. \\[-0.1in]
 \end{enumerate}
 

 \noindent Our algorithm for $t_{\max}$ global rounds is summarized in Alg. 1.

 \begin{algorithm}[t] \label{alg:conn} \small{ 
        \caption{Connectivity-Aware Semi-Decentralized Learning}
        \label{alg:parameter_estimation_algorithm}
        \textbf{Input:} $n$, $c$, $\tloc$, $\phi_{\max}$, $t_{\max}$, $m(0)$, $\{n_\ell(t)\}_{t=0}^{t_{\max}}$, $x^{(0)}$  \\
        \textbf{Output:} $x^{(t_{\max})}$
        \begin{algorithmic}[1]
        \For{$t\in\{0,1,\ldots, t_{\max}-1\}$}
                \State Client $i\in [n]$ sets $x_i^{(t,0)}\leftarrow x^{(t)}$
                \For{$k\in\{0,1,\ldots, \tloc-1\}$}
                    \State Client $i\in[n]$ computes $x_i^{(t,k+1)} \leftarrow x_i^{(t,k)} - \eta_t\sgd f_i(x_i^{(t,k)})$
                \EndFor
                \State \textbf{end}
                \State Client $i\in[n]$ transmits its scaled cumulative local gradient $-\eta_t\sum_{k=0}^{\tloc-1} \sgd f_i(x_i^{(t,k)})= x_i^{(t,\tloc)}-x^{(t)}$ to its out-neighbors $\mathcal N^+_i(t)$
                \State Client $i\in [n]$ computes the following weighted sum of its in-neighbors' cumulative local gradients:
                $$
                    \Delta_i(t)\leftarrow \sum_{j\in\mathcal N_i^{-}(t)} \frac{1 }{d_j^+(t) }\left(x_j^{(t,\tloc)}-x^{(t)} \right)
                $$
                \State PS samples $m_\ell(t) = \frac{n_\ell(t)}{n}m(t)$ clients uniformly at random from cluster $\ell\in[c]$
                \State PS computes $x^{(t+1)}\leftarrow x^{(t)}+ \frac{1}{m(t)}\sum_{i=1}^n\tau_i(t)\Delta_i(t)$ and broadcasts $x^{(t+1)}$ to all clients
                \State PS computes\\ $m(t+1)\leftarrow 
                \min\left\{r\in [n]: \psi(r, \alpha_1(t+1),\ldots, \alpha_c(t+1)) \le \phi_{\max}\right \}$
        \EndFor
        \State \textbf{end}
        \State \Return {$x^{(t_{\max})}$}
        \end{algorithmic}}
\end{algorithm}

\section{Convergence Analysis}
We now provide theoretical performance guarantees for Algorithm 1. We also explain how the effect of D2D cluster connectivity on the convergence rate of the algorithm is captured by the singular values of the equal-neighbor adjacency matrices of the clusters. \textbf{All the proofs except that of Proposition 5.1 are available in the appendix}.

\newtheorem{assumption}{Assumption}

\subsection{Assumptions and Preliminaries}

\subsubsection{Loss Functions} We start by making the following standard assumptions on the local loss functions:

\begin{assumption} [\textbf{Strong Convexity}] \label{asm:strong_convec} All the local loss functions $\{f_i\}_{i=1}^n$ are $\mu$-strongly convex, i.e., there exists $\mu>0$ such that  $\left(\nabla f_i(x)-\nabla f_i(y)\right)^\top\left(x-y\right)\ge \mu\norm{x-y}^2$ for all $x,y\in\R^p$ and all $i\in[n]$.
\end{assumption}


\begin{assumption}[\textbf{Smoothness}]\label{asm:smoothness}  All the local loss functions $\{f_i\}_{i=1}^n$ are $\beta$-smooth, i.e., there exists a finite $\beta$ such that $\norm{\nabla f_i(x)-\nabla f_i(y)}\le \beta\norm{x-y}$ for all $x,y\in\R^p$ and all $i\in [n]$.
\end{assumption}

As shown in~\cite{lin2021semi},  Assumptions~\ref{asm:strong_convec} and~\ref{asm:smoothness} imply that the global loss function $f$ is both $\mu$-strongly convex and $\beta$-smooth.

\subsubsection{SGD Iterations} Additionally, we make the following standard assumption on the stochastic gradients generated through the SGD  procedure for each client:
\begin{assumption} [\textbf{Unbiasedness and Bounded Variance}] The SGD noise associated with every client is unbiased, i.e., \\${\E[\sgd f_i(x)-\nabla f_i(x)\mid x]=0}$, and it has a bounded variance, i.e., there exists a constant $\noise>0$ such that  ${\E\|\sgd f_i(x)-\nabla f_i(x)}\|^2\le \noise^2 $ for all  models $x\in\R^p$ and all $i\in [n]$.
\end{assumption}
In addition, we assume that the SGD noise is independent across clients, i.e., for all $x\in\R^p$, the random vectors $\left\{\sgd f_i(x) - \nabla f_i(x) \right\}_{i=1}^n$ are mutually conditionally independent given $x$.

\subsubsection{Gradient Diversity} Furthermore, we assume that the training data are not distributed uniformly at random among the clients, which gives rise to data heterogeneity 
 among the clients. Unlike the standard assumption on data heterogeneity that imposes a uniform upper bound on $\norm{\nabla f_i(x)-\nabla f(x)}$ (see~\cite{wang2019adaptive} for example), we make a weaker assumption on the diversity of local gradients. In fact, this assumption, which was first proposed in~\cite{lin2021semi}, can be derived as a consequence of Assumptions~\ref{asm:strong_convec} and~\ref{asm:smoothness}, as shown in~\cite{lin2021semi}. Below, we formally state this observation.

\begin{lemma} [\textbf{Gradient diversity }~\cite{lin2021semi}] \label{lem:grad_div}
For all $i\in [n]$ and $x\in\R^p$, we have ${ \|\nabla f_i(x) - \nabla f(x) \| \le \delta + 2\beta\|x - x^*\|}$, where
\begin{align}\label{eq:delta_definition}
    \delta:=\beta\max_{i\in [n]}\|x^*-x_i^*\|=\beta\max_{i\in [n] }\| x^*-\arg\min_{y\in\R^p} f_i(y)\|
\end{align}
\end{lemma}

As argued in~\cite{lin2021semi}, the standard assumption (which is  a special case of the above inequality with $\beta=0$) is unrealistic as it does not apply to quadratic and super-quadratic loss functions unless the upper bound $\delta$ is chosen to be unreasonably large.
\subsection{Results}
We now quantify how the singular values of the equal-neighbor matrices and the number of clients sampled by the PS affect the efficiency of our algorithm in terms of its optimality gap.

We first show how the expected optimality gap of our algorithm depends on the expected deviation of the global average ${x^{(t+1)}-x^{(t)}}$ (i.e., the random vector computed by the PS using the aggregation rule~\eqref{eq:global_aggregation}) from the true average of all the scaled cumulative gradients.

\begin{lemma}\label{lem:prop} At the end of the $(t+1)$-th round of global aggregation, the expected optimality gap of Algorithm 1 is given by
$$
     \E\norm{x^{(t+1)} - x^*}^2 = \E\norm{x^{(t+1)} - \bar x^{(t+1)} }^2 + \E\norm{\bar x^{(t+1)} - x^*}^2,
$$
where $\bar x^{(t+1)} := x^{(t)}+ \frac{1}{n} \sum_{i=1}^n (x_i^{(t,T)} - x^{(t)})$ is a vector that would equal the global model if the PS were to sample all the $n$ clients.
\end{lemma}

Observe that the first term on the RHS depends on $x^{(t+1)}-\bar x^{(t+1)}$, which can be easily shown to be the difference between the random average $\frac{1}{m(t) }\sum_{i\in \mathcal S(t)}\Delta_i(t)$ and the true average $\frac{1}{n}\sum_{i=1}^n \left( x_i^{(t,T)}-x^{(t)}\right)$. Thus, this term captures the error due to random sampling. As the next result shows, this difference depends on the network topology as well as on $m(t)$, the number of clients selected for global aggregation uniformly at random by the PS. 
\begin{proposition}\label{prop:second_major_chunk}
Let $\delta$ be the constant defined in~\eqref{eq:delta_definition}. Then Algorithm~1 satisfies the following for every $t\in\N\cup\{0\}$:
\begin{align*}
    \E\norm{x^{(t+1)}-\bar x^{(t+1)}}^2&\le \bigg(2T\noise^2\eta_t^2 +  4eT(\noise^2+2\delta^2) \eta_t^2 + 6\delta^2 T^2\eta_t^2\cr
    &\quad\quad+(27+4e) T^2\beta^2\eta_t^2 \E\norm{x^{(t)}-x^* }^2\bigg)\phi(t),
\end{align*}
where $\phi(t)$ is the connectivity factor defined in~\eqref{eq:conn_factor}.
\end{proposition}

In other words,  $\E\norm{x^{(t+1)}-\bar x^{(t+1)}}$ depends on the previous optimality gap $\E\norm{x^{(t)} - x^* }^2$ via $\phi(t)$, i.e., the connectivity factor that captures the combined effect of global aggregation  (via $m(t)$) and the D2D network topology within each cluster (via $\phi_\ell(\alpha_\ell(t) )$). 

Moreover, Lemma~\ref{lem:prop} and Proposition~\ref{prop:second_major_chunk} together show that the singular values of the equal-neighbor adjacency matrices can be used to derive an upper bound on the expected optimality gap (and ultimately establish theoretical performance guarantees) for our connectivity-aware algorithm. Doing so yields the following.

\begin{proposition}\label{prop:lots_of_notation_1}
Let $\delta$ be as defined in~\eqref{eq:conn_factor}, let $\phi(t)$ be the connectivity factor defined in~\eqref{eq:conn_factor}, let ${\Gamma:=f(x^*)-\frac{1}{n}\sum_{i=1}^n \min_{x\in\R^p}f_i(x)}$, and let $e$ denote the exponential constant. Then the expected optimality gap of Algorithm 1 satisfies the following for all $t\in\N_0$:
\begin{align*}
    &\E\norm{x^{(t+1)}-x^*}^2\cr &\le \left((1-\mu\eta_t)^T + (27 + 4e)T^2 \beta^2 \eta_t^2(2T + \phi(t))\right) \E\norm{x^{(t)}-x^*}^2\cr
    &\quad+T \left( \frac{\noise^2}{n} + 6\beta\Gamma + 4T\noise^2
         +  8eT(\noise^2+2\delta^2)  + 12\delta^2 T^2\right)\eta_t^2\cr
    &\quad + \left(2T\noise^2 +  4eT(\noise^2+2\delta^2)  + 6\delta^2 T^2\right)\phi(t)\eta_t^2.
\end{align*}
\end{proposition}

A recursive expansion on the inequality 
stated by Proposition~\ref{prop:lots_of_notation_1} results in our main theoretical result, which we state below.

\begin{theorem}\label{thm:main}
Consider a connectivity factor threshold $\phi_{\max}\ge 0$, and suppose that $\phi(t)\le \phi_{\max}$ for all times $t\ge 0$. In addition, suppose $\eta_t=\frac{4}{T\mu(t+t_1)}$, where 
$$
    t_1:=\left\lfloor 4\left(1-\frac{1}{T}\right) + (16T + 8 \phi_{\max})\left(\frac{\beta}{\mu}\right)^2 + 1\right\rfloor.
$$
Then the expected optimality gap of Algorithm 1 satisfies the following for all $t\ge 0$:
\begin{align}
    &\hspace{-2mm} \mathbb{E}\left\|x^{(t)}-x^*\right\|^2 \cr
    & \hspace{-2mm} \leq  \left(\frac{t_1}{t+t_1}\right)^2 \mathbb{E}\left\|x^{(0)}-x^*\right\|^2  +\frac{16\left(\frac{1}{nT}\left(\frac{\noise}{\mu}\right)^2+6\frac{\beta\Gamma}{T\mu^2}\right)}{t+t_1}\cr
    &\hspace{-2mm} +\frac{\left(32T+16\phi_{\max}\right)\left(\frac{2}{T}\left(\frac{\noise}{\mu}\right)^2 + \frac{4e}{T}\left(\left(\frac{\noise}{\mu}\right)^2 + 2\left(\frac{\delta}{\mu}\right)^2 \right) + 6  \left(\frac{\delta}{\mu}\right)^2 \right)}{t+t_1}.
\end{align}
\end{theorem}

Theorem~\ref{thm:main} reveals that the convergence rate of our algorithm is $\mathcal{O}\left({1}/{t} \right)$, which coincides with that of FedAvg and its semi-decentralized variants such as~\cite{lin2021semi}. In fact, $\mathcal{O}\left({1}/{t} \right)$ resembles the convergence rate of vanilla centralized SGD.
It also shows that suitably tuning the connectivity factor (by choosing an appropriate value of $\phi_{\max}$) is critical to the efficiency of the algorithm: as $\phi_{\max}$ increases the bound gets worse/larger; however, $\phi_{\max}$, by its definition, is non-negative, which means it can at best be made equal to 0, which forces $m=n$, in which case the inequality boils down to an upper bound on the convergence rate of FedAvg with full device sampling. At the other extreme, setting $\phi_{\max}$ to $\infty$ results in $m=1$, which happens when our semi-decentralized FL architecture collapses to full decentralization.

Moreover, Theorem~\ref{thm:main} jointly captures the effect of the following factors on the expected instantaneous optimality gap and hence on the convergence rate: (i) the initial optimality gap $\E\|x^{(0)}-x^*\|^2$ (via the first term), (ii) The SGD noise variance $\noise^2$ and the strong convexity $\mu$ and smoothness parameters $\beta$ (via the second term), and finally, (iii) the combined effect of cluster connectivity levels and random sampling-based global aggregations (via the third term, which depends on $\phi_{\max}$, which in turn prevents the connectivity factor $\phi(t)$ from becoming too large).
It can be seen that higher values of the SGD noise variance $\noise^2$ and the data heterogeneity measure $\Gamma$ lead to a larger value of the bound, implying that our algorithm is sensitive to the size of the mini-batches used for computing the stochastic gradients as well as to the non-i.i.d.-ness of the local datasets.

\section{Singular Value Bounds}\label{sec:sing}
Having established the role of the connectivity factor $\phi(t)$ in the performance of Algorithm 1, we now analyze two important quantities associated with $\phi(t)$: the top two singular values of the equal-neighbor adjacency matrices of the clusters. Since a precise estimation of these singular values requires full knowledge of the cluster topologies, which is challenging to obtain in practice, we are motivated to derive a set of novel upper bounds on these values in terms of the node degrees of the cluster digraphs, which are easy to obtain/measure in practice. To the best of our knowledge, this is one of the first attempts at connecting the singular values of adjacency matrices with minimal topological information such as node degrees of the digraphs.

To conduct our analysis, 
for any digraph $G=([s],E)$, we first define
$
    \varepsilon=\varepsilon_G:= \frac{\dmax(G)-\dmin(G) }{\dmin(G)},
$
which  quantifies the heterogeneity of out-degree of the nodes across the digraph. We also let $\alpha(G):=\frac{\dmin(G)}{s}$ capture the minimum fraction of the node population that any node is out-connected to. In addition, we let $W(G)=(w_{ij})$ and $D^+(G):=\text{diag}([d_1^+\,\,\, d_2^+\,\,\,\cdots\,\,\, d_s^+]^\top)$ denote the binary adjacency matrix and the out-degree matrix of $G$, respectively. In the sequel, we drop the indexing $(G)$ for brevity.

\def\dinmax{d_{\max}^{-}} 
We are now equipped to state our first set of bounds on the greatest two singular values of $G$ under certain regularity assumptions on the digraph. 
\begin{proposition}\label{prop:sing_bounds}
Suppose $G = ([s],E)$ is a directed graph in which every node has its in-degree equal to its out-degree, i.e., $d_i^+ = d_i^-$ for all $i\in[n]$. Then the greatest two singular values $\sing_1$ and $\sing_2$ of the  equal-neighbor adjacency matrix $A$ of $G$ satisfy the following inequalities for $\alpha>\frac{1}{2}$ and $\varepsilon\ll 1$: 
\begin{align}
    \sing^2_1&\le 1+\varepsilon + \mathcal{O}(\varepsilon^2),\label{eq:first_sing_bound}\\
    \sing_2^2&\le \left(\frac{1}{\alpha}-1 \right)^2 +2\varepsilon\left(1 + \frac{2}{\alpha} - \frac{1}{\alpha^2} \right) + \mathcal{O}(\varepsilon^2),\label{eq:second_sing_bound}
\end{align}
where $\mathcal{O}(\cdot)$ is the big-$O$ notation used in the context of $\varepsilon\to 0$.
\end{proposition}
\begin{proof}
To simplify our notation, we define $D:=D^+$ for the remainder of this proof. Observe that
$$
    A^\top=D^{-1}W = D^{-\frac{1}{2}}( D^{-\frac{1}{2}} W D^{-\frac{1}{2}} )D^{\frac{1}{2}},
$$
which means $A^\top$ is similar to the normalized adjacency matrix defined as $A_N:=D^{-\frac{1}{2}} W D^{-\frac{1}{2}}$. 

On the other hand, we have $\dmin=\dmax(1-\varepsilon)\le d_i^+\le\dmax$ for all $i\in[s]$, which implies the existence of a diagonal matrix $E_3$ such that $O\le E_3\le I$ and ${D = \dmax\left((1-\varepsilon)I + \varepsilon E_3\right)}$. Using similar arguments, it can be easily shown that there exist diagonal matrices $E_1$ and $E_2$ such that $O\le E_1,E_2\le I$, $D^{\frac{1}{2} }=\sqrt{\dmax}\left((1-\frac{\varepsilon}{2})I + \frac{\varepsilon}{2} E_1 \right) + \mathcal{O}(\varepsilon^2)$, and ${D^{-\frac{1}{2}}=\frac{1}{\sqrt{\dmax}}\left(I+\frac{\varepsilon}{2} E_2\right)+\mathcal{O}(\varepsilon^2)}$. As a result, the following holds up to an additive error of $\mathcal{O}(\varepsilon^2)$:
\begin{align*}
    A_N = D^{\frac{1}{2}}A^\top D^{-\frac{1}{2}} &= \sqrt{\dmax}\left((1-\frac{\varepsilon}{2})I + \frac{\varepsilon}{2} E_1 \right)\frac{A}{\sqrt{\dmax}}\left(I+\frac{\varepsilon}{2} E_2\right)  \cr
    &=A^\top + \frac{\varepsilon }{2} \left( (E_1-I)A+AE_2 \right),
\end{align*}
i.e., $A^\top-A_N =- \frac{\varepsilon }{2} \left( (E_1-I)A^\top+A^\top E_2 \right) + \mathcal{O}(\varepsilon^2)$. In conjunction with standard bounds on singular value perturbations (e.g., see~\cite{meyer2000matrix}), this implies the following up to an additive error of $\mathcal{O}(\varepsilon^2)$:
\begin{align}\label{eq:sing_perturb}
    \sing_j(A) &= \sing_j(A^\top)\cr
    &\le \sing_j(A_N) + \frac{\varepsilon}{2}\norm{(E_1-I)A^\top+A^\top E_2 }\cr
    &\le \sqrt{\lambda_j(A_N A_N^\top)} + \frac{\varepsilon}{2}\left(\norm{E_1-I}\norm{A^\top}+\norm{A^\top}\norm{E_2}\right)\cr
    &\stackrel{(a)}\le \sqrt{\lambda_j( D^{-\frac{1}{2}}WD^{-1}W^\top D^{-\frac{1}{2}})} + \frac{\varepsilon}{2}( 1\cdot\norm{A^\top}+\norm{A^\top}\cdot 1)\cr
    &\stackrel{(b)}=\sqrt{\lambda_j(D^{-1}WD^{-1}W^\top)} + \varepsilon\sing_1(A^\top).
\end{align}
Here, $(a)$ holds because $I-E_1$ being a diagonal matrix along with $O\le I-E_1\le I$ implies that ${\norm{I-E_1}=\max_{i\in[s]}\left|(I-E_1)_{ii}\right|\le 1}$. $(b)$ holds because $\sing_1(A^\top)=\norm{A^\top}$ and because $D^{-1}WD^{-1}W^\top$ and $D^{-\frac{1}{2}}WD^{-1}W^\top D^{-\frac{1}{2}}$, being similar, have the same eigenvalues.

We now bound $\sing_1(A)$ and $\sing_2(A)$ individually. As for $\sing_1(A)$, the derivation~\eqref{eq:sing_perturb} and the fact that $\sing_1(A)=\sing_1(A^\top)$ imply that \begin{align}\label{eq:first_sing}
    \sing_1(A) \le \frac{\sqrt{\lambda_1(D^{-1}W D^{-1}W^{T} )} }{ 1 - \varepsilon} &= \sqrt{\lambda_1(D^{-1}W D^{-1}W^{T})}(1+\varepsilon) \nonumber\\
    &\quad+ \mathcal{O}(\varepsilon^2).
\end{align}
So, it is enough to bound $\lambda_1(D^{-1}W D^{-1}W^{T} )$. For this purpose, note that $A$ being column-stochastic implies that $D^{-1}W\allone=\allone$ and hence also that $W\allone = D\allone$. Besides, our assumption on in-degrees and out-degrees can be expressed as $\sum_{j=1}^sw_{ij}=\sum_{j=1}^s w_{ji}$ for each $i\in[s]$, or equivalently, $W^\top\allone=W\allone = D\allone$. As a result, we have $D^{-1}W^\top\allone=\allone$. Thus, $D^{-1}WD^{-1}W^\top = A^\top D^{-1}W^\top$ is a product of row-stochastic matrices and hence, it is row-stochastic in itself. Consequently, $\lambda_1(D^{-1}W D^{-1}W^{T} )=1$. In light of this,~\eqref{eq:first_sing} implies~\eqref{eq:first_sing_bound}.

It remains to prove~\eqref{eq:second_sing_bound}. We do this by using Theorem 2.2 of~\cite{lynn1969bounds}, which helps derive a bound in terms of $\sing_1$ and the minimum positive entry $\delta$ of the matrix $D^{-1}W D^{-1}W^{T}$. We first note that  
\begin{align*}
    &(D^{-1}W D^{-1}W^{T})_{ij} \stackrel{(a)}\ge \frac{1}{({\dmax})^2}\sum_{k=1}^s (W)_{ik}(W^\top)_{kj}\cr
    &= \frac{1}{({\dmax})^2}|\{k\in [s]: w_{ik}=w_{jk}=1\}| \stackrel{(b)}\le \frac{(2\alpha-1)s}{ ({\dmax})^2},
\end{align*}
where $(a)$ follows from the fact that $D^{-1}\ge \frac{1}{\dmax}I$ and $(b)$ holds because the number of common out-neighbors of any two nodes $i,j\in [s]$ is at least $(2\alpha-1)s$. We can now apply Theorem 2.2 of~\cite{lynn1969bounds} by setting $x=\frac{1 }{ \sqrt{s} }$ in the theorem (because $\frac{1}{\sqrt{s}}\allone$, as explained above, is the unit-norm principal eigenvector of $D^{-1}W D^{-1}W^\top$). Thus,
\begin{align}\label{eq:alpha_bound}
    \lambda_2(D^{-1}W D^{-1}W^\top)&\le \lambda_1(D^{-1}W D^{-1}W^\top) - \frac{(2\alpha-1)s^2}{{(\dmax)}^2}\cr
    &= 1 -\left(\frac{2}{\alpha} - \frac{1}{\alpha^2}\right)(1-2\varepsilon) + \mathcal{O}(\varepsilon^2),
\end{align}
where the last step follows from the observation that $\dmax=\frac{\alpha s}{1-\varepsilon}$. 

Combining~\eqref{eq:first_sing_bound},~\eqref{eq:sing_perturb} and~\eqref{eq:alpha_bound} now gives
\begin{align*}
    \sing_2(A)\le \sqrt{1 - \left(\frac{2}{\alpha} - \frac{1}{\alpha^2}\right)(1-2\varepsilon) +\mathcal{O}(\varepsilon^2)} + \varepsilon\left(1+\varepsilon+\mathcal{O}(\varepsilon)^2\right).
\end{align*}
Squaring both sides and rearranging the terms results in~\eqref{eq:second_sing_bound}.
\end{proof}
\newtheorem{remark}{Remark}

\begin{remark}
Observing the bounds in Proposition~\ref{prop:sing_bounds},  we can see that setting $\alpha =1$, which corresponds to $G$ being a clique, in the  bounds yields $\sing_1\le 1+\mathcal{O}(\varepsilon)$ and $\sing_2=\mathcal{O}(\varepsilon)$. These inequalities, for $\varepsilon\ll 1$, are tight with respect to the well-known lower bounds  $\sing_1\ge 1$ and $\sing_2\ge 0$. This implies that the bounds~\eqref{eq:first_sing_bound} and~\eqref{eq:second_sing_bound} can be expected to be reasonably tight for high edge density (i.e., whenever $\alpha\approx 1$). Another implication of the bounds is decreasing $\varepsilon$, which inherently measures how irregular the digraph is, leads to~\eqref{eq:second_sing_bound} becoming sharper. 
\end{remark}
\def\dinmax{d^{\text{in}}_{\max}}
\def\dmax{d_{\text{max}}}

The above singular value bounds are especially tight for digraphs that are approximately regular (or digraphs that do not exhibit significant variations in their in-degrees and out-degrees) since such digraphs satisfy $\varepsilon\ll1$. This happens in practice, when the D2D clusters are dense (e.g., in the wireless setting, when the nodes are closer to each other or when they can move and communicate over time).
Furthermore, the same holds for the condition on $\alpha$ in Proposition~\ref{prop:sing_bounds} (i.e., $\alpha>\frac{1}{2}$), which is always met when the clusters are dense.

However, the bounds~\eqref{eq:first_sing_bound} and~\eqref{eq:second_sing_bound} are obtained under the assumption  that every node has its in-degree equal to its  out-degree, which can be restrictive in practical settings. This observation further motivates us to find a new set of singular value bounds that work well under milder assumptions. We thus provide the following bounds, which not only relax the said restrictive assumption, but also apply to digraphs with more general out-degree distributions (and hence subsume digraphs with wider out-degree variations). 
\begin{proposition} \label{prop:lots_of_notation} Let $\varphi=\frac{d^{\text{in}}_{\max} - d_{\min}}{d_{\min}}$, where $\dinmax$ denotes the maximum in-degree of the digraph $G$. If $\alpha\ge \frac{1}{2}$, we have the following bounds:
\begin{align}
    \sing_1^2&\le 1+\varphi,\label{eq:general_first_sing_bound}\\
    \sing_2^2&\le 1+\varphi-\frac{(1-\varepsilon)^2(1-\alpha_{-1}^2)\left((1-\varepsilon)^2(1-\alpha_{-1}^2)-\alpha_{-1}\right)}{s(\varepsilon_{\text{net}}+1)\left(\varepsilon_{\text{net}} - \alpha_{-1} + \frac{1}{\alpha s} \right )}, \label{eq:general_second_sing_bound}
\end{align}
where $\varepsilon_{\text{net}}:=\varphi+\frac{\varepsilon}{\alpha}$ and $\alpha_{-1}:=\frac{1}{\alpha}-1$.
\end{proposition}
The bounds obtained in Proposition~\ref{prop:lots_of_notation} (\textbf{proved in the appendix}) are particularly effective when the D2D cluster digraphs are dense but irregular. This is often the case in practical systems, when there is communication heterogeneity (e.g., in wireless sensor networks consisting of sensors with different radii).

In conjunction with Theorem~\ref{thm:main}, the bounds derived in Propositions~\ref{prop:sing_bounds} and~\ref{prop:lots_of_notation} capture the inherent  dependence of the expected optimality gap, and hence that of the convergence rate, on the degree distributions of the D2D clusters. In particular, upon having approximately regular D2D clusters, the bounds in Propositions~\ref{prop:sing_bounds} along with the result of Theorem~\ref{thm:main} determine the convergence rate of Algorithm 1. The same holds when using the result of Proposition~\ref{prop:lots_of_notation} with Theorem~\ref{thm:main}, which will characterize the convergence rate of Algorithm 1 upon having irregular D2D clusters.

\section{Numerical Validation} 
We now conduct numerical experiments to validate our methodology. Overall, our simulations show that compared with baselines, Algorithm 1 obtains significant reductions in total communication cost for the same or similar levels of testing accuracy.

\subsection{Implementation}

\subsubsection{Network Architecture}
We simulate a network consisting of $n=70$ edge devices partitioned into $c=7$ clusters with $n_\ell = 10$ nodes per cluster. In every global aggregation round, the digraph for each cluster $\ell\in[c]$ is constructed as follows: (i) we generate a $k$-regular directed graph (a digraph in which every node has its in-degree and out-degree equal to $k$) with the value of $k$ being chosen uniformly at random from the set $\{6,\ldots, 9\}$; (ii) we delete a fraction $p\in(0,1)$ of the directed edges uniformly at random so as to incorporate D2D link failures due to client mobility and bandwidth issues. The result is an approximately regular digraph whose degree distribution may deviate significantly from that of regular digraphs, while satisfying $\alpha_\ell(t)>\frac{1}{2}$.

\subsubsection{Datasets}  All our simulations are performed on  MNIST~\cite{yan1998mnist} and Fashion-MNIST
(F-MNIST)~\cite{xiao2017fashion} datasets. The MNIST dataset consists of 70K images (60K
for training and 10K for testing), and each image is a  hand-written digit between 0 to 9 (i.e., the dataset has 10 labels). The same applies to the FMNIST dataset, the only difference being that it consists of images of fashion products.

\subsubsection{ML Models and Implementation} We use the neural network model from~\cite{mcmahan2017communication} in our simulations. In particular, we use a convolutional neural network (CNN) with two $5\times 5$ convolution layers, the first of which has 32 channels and the second 64 channels, where each of these layers precedes a $2\times2$ max pooling, resulting in a total model dimension of $1,663,370$. We use the PyTorch implementation of this setup provided in~\cite{ji2018pytorch} with cross-entropy loss. Each dataset is distributed among the clients in a non-i.i.d. manner: the samples (from either of the two datasets) are first sorted by their labels, partitioned into chunks of equal size, and each of the 70 clients is assigned only two chunks (i.e., each client will end up having only two labels). This results in extreme data heterogeneity, which leads to strong empirical guarantees for our approach. 


All of our simulations are performed using the following hyperparameter values/ranges: $\tloc = 5$, $t_{\max}\in\{15,30\}$, $p\in\{0.1,0.2\}$, and $\eta_t = 0.02(0.1)^t$ where $t$ is the global aggregation index.

\subsection{Results}
We compare the energy vs. accuracy trade-offs associated with Algorithm 1 with those associated with two baselines, FedAvg~\cite{mcmahan2017communication} and collaborative relaying (COLREL)~\cite{yemini2022semi}. The second baseline is a recently proposed semi-decentralized FL algorithm that incorporates single-step consensus updates. Under the D2D and D2S connectivity constraints introduced in Section~\ref{sec:prob_form}, COLREL is a variant of FedAvg that incorporates one round of column-stochastic D2D aggregations before every global aggregation round but does not provide any criterion to control the sampling size $m$, which we assume to be fixed throughout its implementation. The fundamental difference between our method and COLREL is that our method takes into account the change in the connectivity of D2D clusters, optimally tuning the value of $m$ according to the set of novel upper bounds on the singular values we obtained in Section~\ref{sec:sing}.

We consider these tradeoffs under different D2S connectivity levels. Intuitively speaking, on one hand, as the D2S connectivity improves, we expect to see that our algorithm leads to a lower energy and cost savings as compared to FedAvg. This is because our algorithm will naturally collapse to FedAvg and D2D communications will become less useful since more devices would engage in uplink communications, which by itself degrades the benefit of D2D local aggregations. On the other hand, as D2S connectivity improves, we expect to see that our algorithm achieves significant energy savings as compared to COLREL. This is because the impact of tuning $m$ becomes more prominent when there is a possibility of D2S communications.

All of the following plots and discussion are based on the assumption that  the ratio of the energy required for D2D communication to that of up-link (D2S) transmission, denoted by $\frac{E_{\text{D2D}}}{E_{\text{Glob}}}$, equals $0.1$. This is a pessimistic estimate in favor of D2S considering that most ratios reported in the literature~\cite{zhang2017security,hmila2019energy,lin2021semi} take values less than $0.1$. Thus, the communication costs reported are $(\#\text{D2S transmissions}) + 0.1 \times (\#\text{D2D transmissions})$.


\begin{figure}[t]\label{fig:52_57}
    \centering
    \includegraphics[scale=0.195] {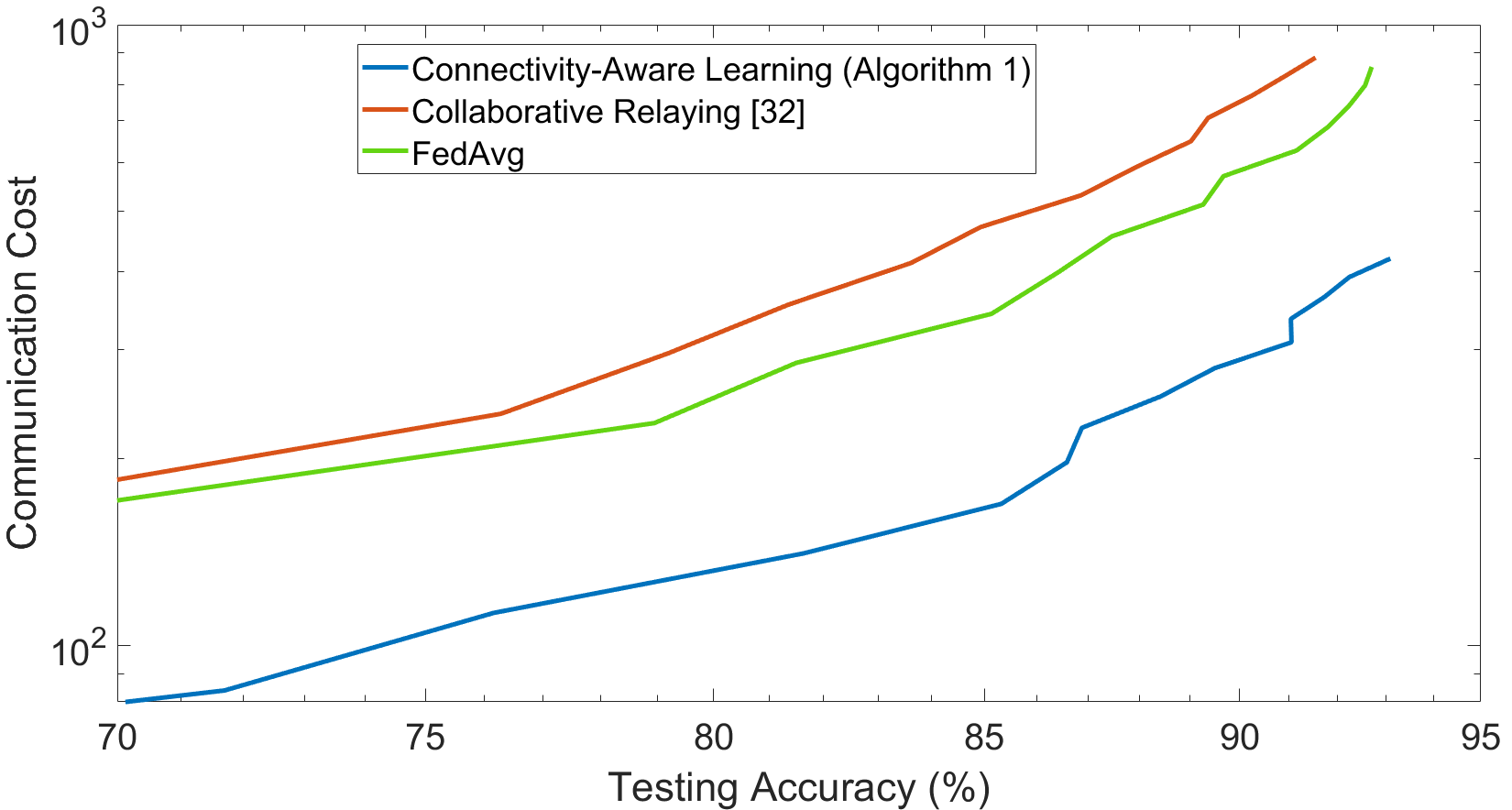}
    \caption{Communication cost vs. testing accuracy under high D2S connectivity (Dataset: MNIST).}
    \label{fig:my_label}
\end{figure}
\begin{figure}[t]\label{fig:fmnist_52_57}
    \centering
    \includegraphics[scale=0.195]{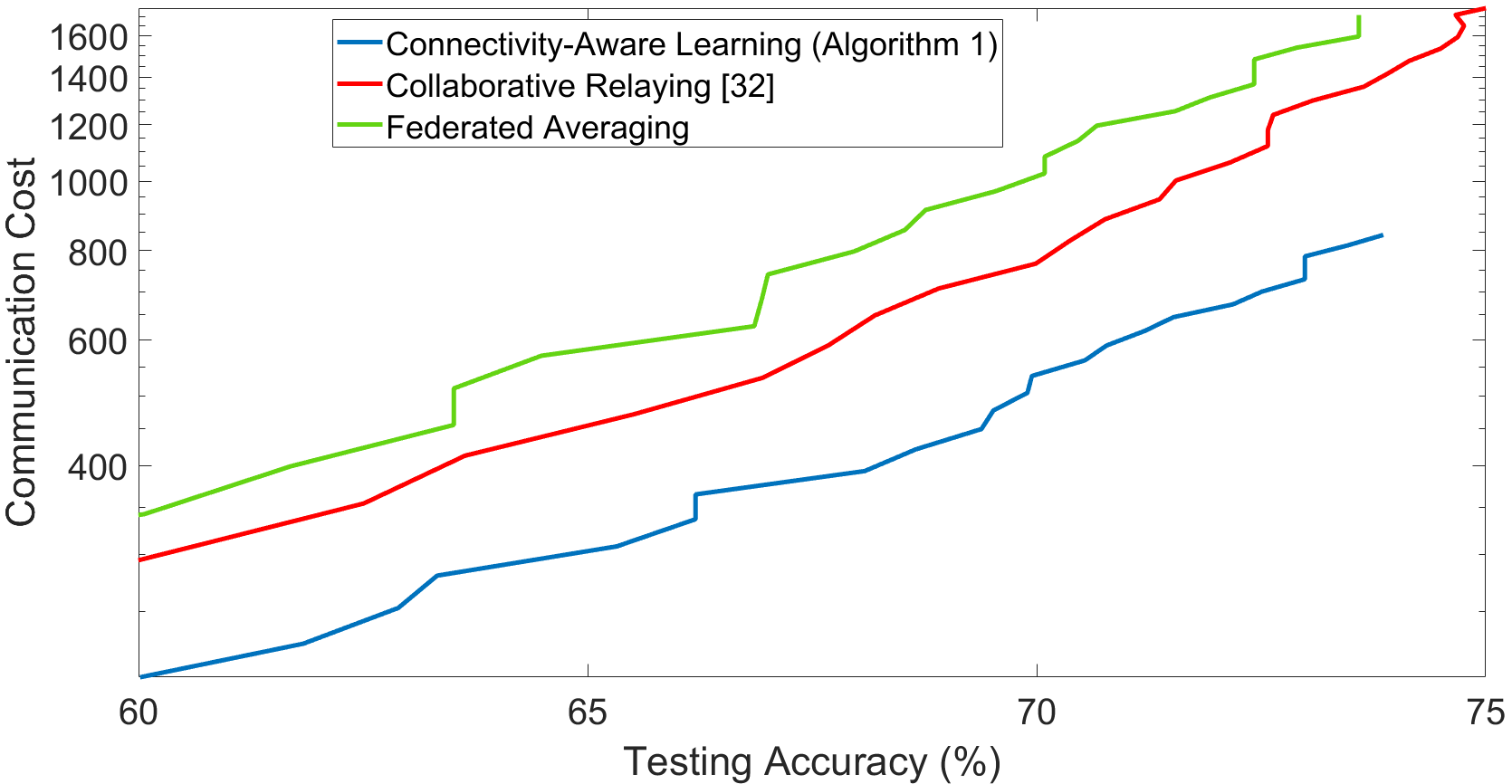}
    \caption{Communication cost vs. testing accuracy under high D2S connectivity (Dataset: F-MNIST).}
    \label{fig:my_label}
\end{figure}

\subsubsection{Case 1: Cost savings under high D2S connectivity and a low link failure probability}
When the PS has a high downlink bandwidth and the connectivity between the devices and the PS is reliable, implementing FedAvg or COLREL has the effect of setting $m$ to a value close to $n$. As an example, we implement FedAvg and COLREL with $m=57$ and $m=52$, respectively (note that COLREL requires fewer up-link transmissions because it uses D2D consensus updates in addition to global aggregations). The results for MNIST are shown in Fig. 2: choosing $\phi_{\max}=0.06$ and a low D2D link failure probability $p=0.1$ results in Algorithm 1 achieving a testing accuracy of $90$\% while consuming about $46\%$ less energy than FedAvg (thereby incurring proportionately lower communication costs). With respect to COLREL, the energy saving is even higher because COLREL also expends energy on D2D aggregations with relatively little gain in testing accuracy.

Repeating this experiment on FMNIST results in a similar performance, depicted in Fig. 3. We see that Algorithm 1 (with $\phi_{\max}=0.06$) consumes about $30$\% less energy than COLREL for achieving a testing accuracy of 70\%. 


\begin{figure}[t]\label{fig:26_15}
    \centering\includegraphics[scale=0.195]{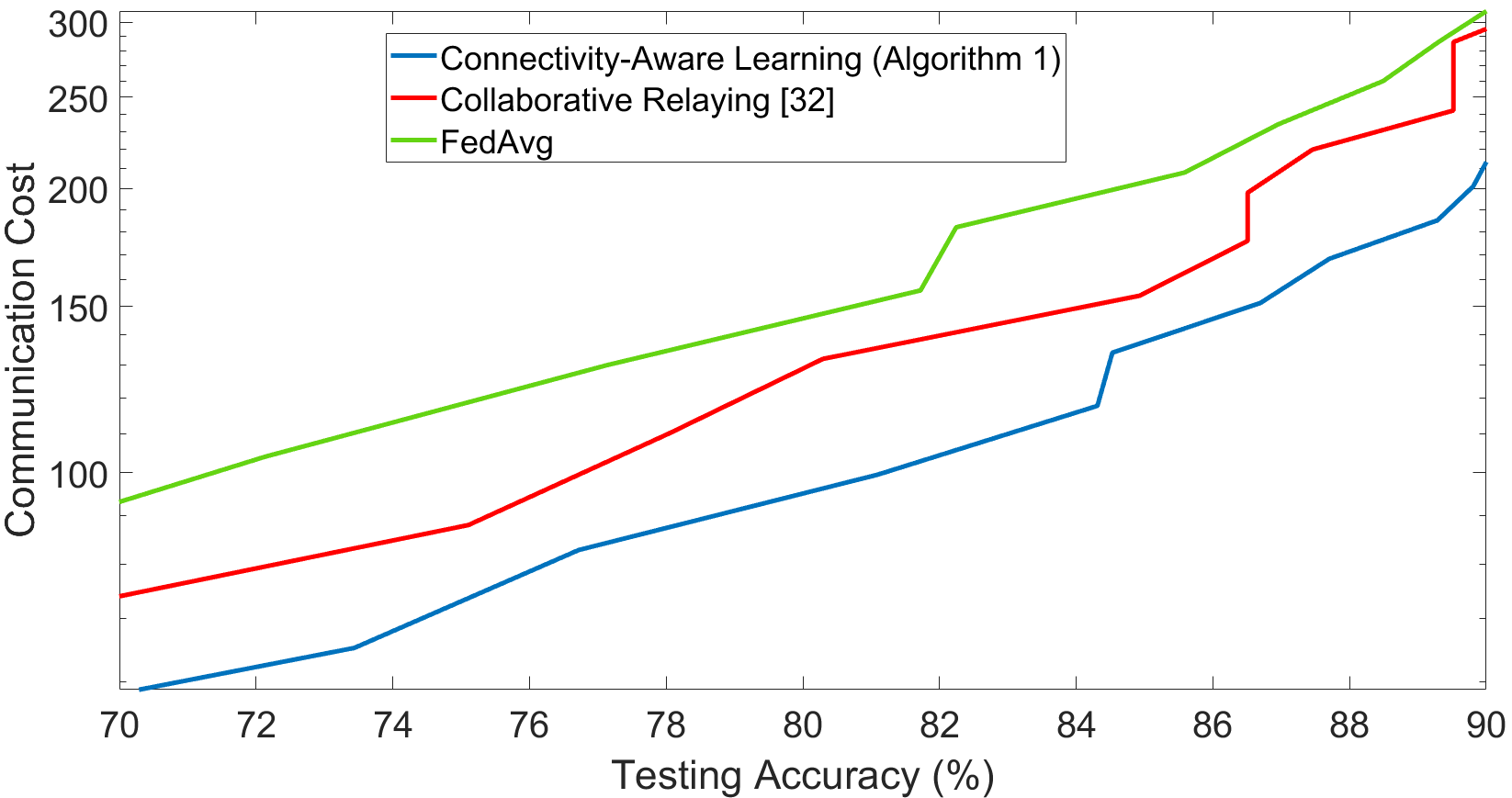}
    \caption{Communication cost vs. testing accuracy under low D2S connectivity (Dataset: MNIST).}
    \label{fig:my_label}
\end{figure}

\begin{figure}[t]\label{fig:32_point_one}
    \centering
    \includegraphics[scale=0.195]{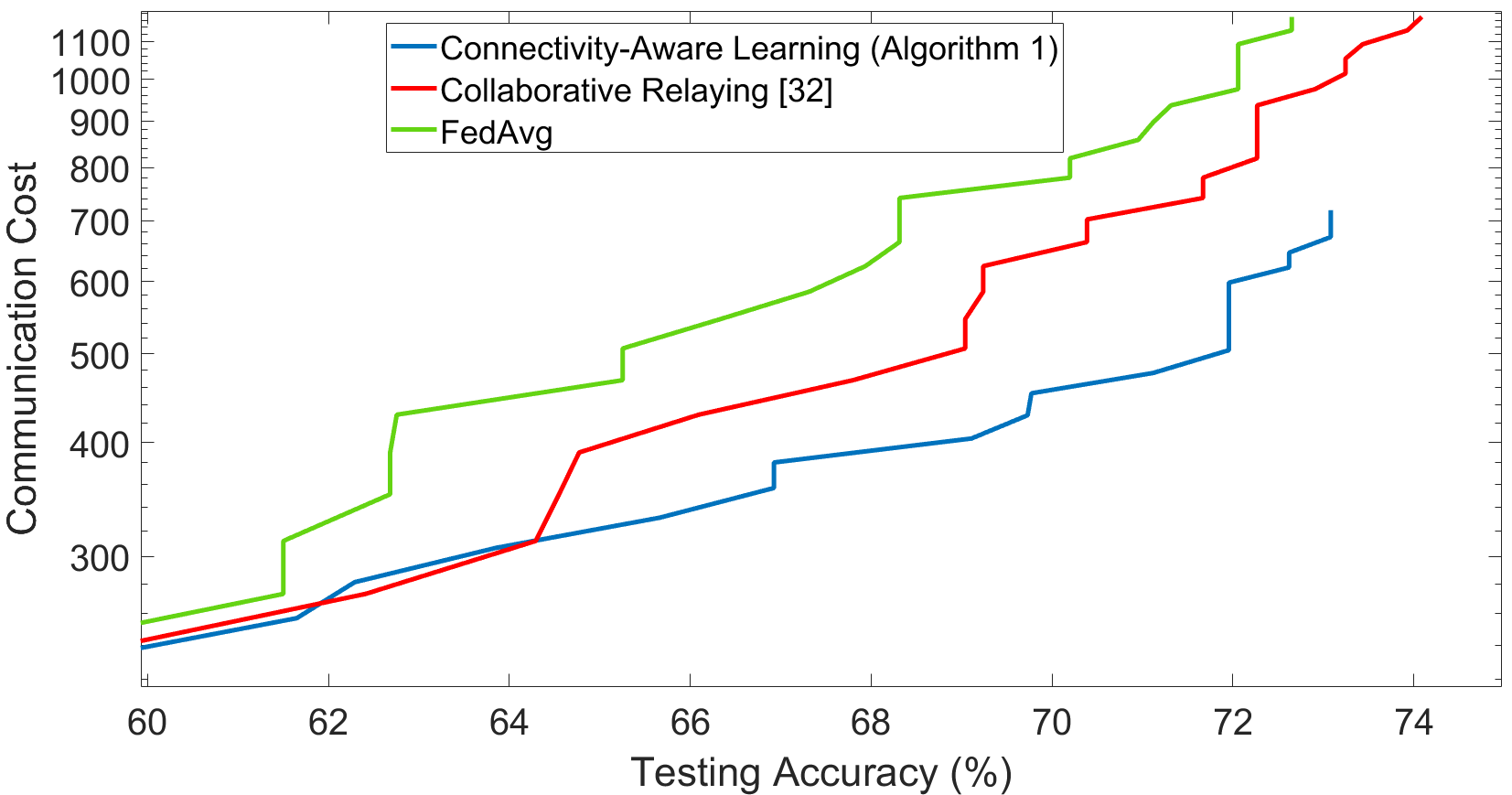}
    \caption{Communication cost vs. testing accuracy under low D2S connectivity (Dataset: F-MNIST).}
    \label{fig:my_label}
\end{figure}

\subsubsection{Case 2: Cost savings under low D2S connectivity and a high link failure probability}
When the connectivity between the devices and the PS is poor, implementing FedAvg or COLREL has the effect of setting $m$ to a value significantly smaller than $n$. As an example, we implement FedAvg and COLREL with $m=26$ and $m=15$, respectively. Choosing $\phi_{\max}=0.2$ and a high D2D link failure probability $p=0.2$ results in Algorithm 1 consuming about $30\%$ less energy than FedAvg for achieving a testing accuracy of $90$\% on MNIST, as shown in Fig.~4. The cost saving is lower than in Case 1 as we expect because the singular value bounds incorporated by our algorithm into its choice of $m(t)$ are looser for higher values of the link failure probability $p$. Repeating this experiment on FMNIST results in a similar performance, depicted in Fig.~5.

\section{Conclusion}
We have investigated consensus-based semi-decentralized learning over clustered D2D networks modeled as time-varying digraphs. We first revealed the connection between the singular values of the column-stochastic matrices used for D2D model aggregations and 
the convergence rate of the learning process. We then derived a set of novel upper bounds on these singular values in terms of the degree distributions of the cluster digraphs, and we used the resulting bounds to design a novel connectivity-aware FL algorithm that enables the central parameter server to tune the number of up-link transmissions by using its knowledge of the time-varying degree distributions of clusters. Our algorithm maintains a continuous balance between the number of model aggregations occurring at the server and those occurring over the edge network, thereby enhancing the resource efficiency of the learning process without compromising convergence.

Future works include obtaining upper bounds on singular values under more general assumptions, and obtaining optimal device sampling schemes for irregular clusters.

\bibliographystyle{ACM-Reference-Format}
\bibliography{bib}

\newpage

\end{document}